\documentclass[a4paper]{article}

\usepackage{amsmath,amsthm,amssymb}

\newtheorem{theorem}{Theorem}
\newtheorem{lemma}[theorem]{Lemma}
\newtheorem{corollary}[theorem]{Corollary}  
\newtheorem{remark}[theorem]{Remark}

\usepackage{listings}

\usepackage{environ}
\usepackage{mathdots}

\newcommand{\repeattheorem}[1]{%
  \begingroup
  \renewcommand{\thetheorem}{\ref{#1}}%
  \expandafter\expandafter\expandafter\theorem
  \csname reptheorem@#1\endcsname
  \endtheorem
  \endgroup
}

\NewEnviron{reptheorem}[1]{%
  \global\expandafter\xdef\csname reptheorem@#1\endcsname{%
    \unexpanded\expandafter{\BODY}%
  }%
  \expandafter\theorem\BODY\unskip\label{#1}\endtheorem
}

\usepackage{tikz}

\tikzstyle{normalNode}=[draw=black, fill=black, circle, minimum size=1.35em, inner sep=0.5mm, scale=0.75]
\tikzstyle{safeNode}=[normalNode, fill=white, node contents={$\bullet$}]
\tikzstyle{labeledNode}=[normalNode, fill=white, minimum size=1.35em, inner sep=0.5mm]
\tikzstyle{normalEdge}=[black, thick, >=stealth]
\tikzstyle{secondaryEdge}=[normalEdge, dashed]

\usepackage[capitalise]{cleveref}

\def\mathrlap{\mathpalette\mathrlapinternal} 
\def\mathrlapinternal#1#2{\rlap{$\mathsurround=0pt#1{#2}$}}
\def\mathllap{\mathpalette\mathllapinternal} 
\def\mathllapinternal#1#2{\llap{$\mathsurround=0pt#1{#2}$}}

\newcommand{\elsum}[1]{\sum_{\mathrlap{#1}}\;}

\usepackage{xparse}
\DeclareDocumentCommand{\rcap}{ O{} m m }{%
	\ifx&#1&%
	\bar{u}_{#2}%
	\else
	\bar{u}_{#2,#1}%
	\fi%
	\ifx&#3&%
	\else
	(#3)%
	\fi%
}
\newcommand{\rcapEq}[3][]{\bar{u}_{#2,#1}(#3)}

\newcommand{\head}[1]{\operatorname{head}(#1)}
\newcommand{\tail}[1]{\operatorname{tail}(#1)}

\newcommand{\outarcs}[1]{\delta^{+}(#1)}
\newcommand{\val}[1]{\operatorname{val}(#1)}
\newcommand{\capacity}[1]{\operatorname{cap}(#1)}
\newcommand{\paths}[1][]{\mathcal{P}_{#1}}
\newcommand{\pathsVia}[2]{\mathcal{P}_{#1 \rightarrow #2}}
\newcommand{\suchthat}{\,:\,}
\newcommand{\NP}{$N\!P$}

\newcommand{\supp}[1]{\operatorname{supp}(#1)}
\newcommand{\NetRemoveArc}[1]{(V, A \setminus \{#1\})}
\newcommand{\NetRemoveSet}[1]{(V, A \setminus #1)}
\newcommand{\Rpaths}[1]{\mathcal{R}(#1)}
\newcommand{\maxreroute}{\textsc{Max RF}}
\newcommand{\maxstrictreroute}{\textsc{Max SRF}}
\newcommand{\minrcut}{\textsc{Min R-Cut}}
\newcommand{\forbiddenpairs}{\textsc{Forbidden Pairs $s$-$t$-Path}}
\newcommand{\rcut}{R-cut}
\newcommand{\LPstrict}{[\textup{LP}_{\textup{strict}}]}
\newcommand{\LPstrictD}{[\textup{D-LP}_{\textup{strict}}]}
\newcommand{\halfint}{half-integral}

\title{Rerouting flows when links fail\footnote{An extended abstract of this work appears in the proceedings of ICALP~2017.}}

\author{Jannik Matuschke\thanks{TUM School of Management and Department of Mathematics, Technische Universit\"at M\"unchen} \and S.\ Thomas McCormick\thanks{Sauder School of Business, University of British Columbia} \and Gianpaolo Oriolo\thanks{Dipartimento di Ingegneria Civile e Ingegneria Informatica, Universit\`a di Roma ``Tor Vergata''}}

\date{}

\begin{document}

\maketitle

\begin{abstract}
We introduce and investigate \emph{reroutable flows}, a robust version of network flows in which link failures can be mitigated by rerouting the affected flow.
Given a capacitated network, a path flow is reroutable if after failure of an arbitrary arc, we can reroute the interrupted flow from the tail of that arc to the sink, without modifying the flow that is not affected by the failure. Similar types of restoration, which are often termed ``local'', were previously investigated in the context of network design, such as min-cost capacity planning. In this paper, our interest is in computing maximum flows under this robustness assumption. An important new feature of our model, distinguishing it from existing max robust flow models, is that no flow can get lost in~the~network.

We also study a tightening of reroutable flows, called \emph{strictly reroutable flows}, making more restrictive assumptions on the capacities available for rerouting.
For both variants, we devise a reroutable-flow equivalent of an $s$-$t$-cut and show that the corresponding max flow/min cut gap is bounded by $2$.
It turns out that a strictly reroutable flow of maximum value can be found using a compact LP formulation, whereas the problem of finding a maximum reroutable flow is \NP-hard, even when all capacities are in $\{1, 2\}$. 
However, the tightening can be used to get a~$2$-approximation for reroutable flows.
This ratio is tight in general networks, but we show that in the case of unit capacities, every reroutable flow can be transformed into a strictly reroutable flow of same value.
While it is \NP-hard to compute a maximal integral flow even for unit capacities, we devise a surprisingly simple combinatorial algorithm that finds a half-integral strictly reroutable flow of value $1$, or certifies that no such solutions exits. 
Finally, we also give a hardness result for the case of multiple arc failures.

 \end{abstract}
 
 \section{Introduction}

	Network infrastructures for transportation, communication, or energy transmission are an important backbone of our society. However, they are also prone to failure or intentional sabotage, and in such cases it is desirable to quickly recover the service provided through the network. A crucial frequent requirement of actual network restoration techniques is that restoration is handled locally~\cite{grandoni2010stable}.  As a motivating example, consider a communication network in which data packets are routed along paths. When a link in the network fails, it is desirable to only reroute the traffic that is actually affected by the failure, i.e., those paths that traverse the failing link, without changing or rerouting any part of the flow that is not affected by the failure. Note that arbitrary rearrangement of the flow after a failure is in general more powerful, but it is both undesirable to interrupt customer service and hard to do so reliably and safely~\cite{sousa2007improving,metcalfe1976ethernet}.

	To cope with such a situation, we introduce the concept of reroutable network flows: A flow on $s$-$t$-paths is \emph{reroutable} if after failure of any arc $\bar{a} = (\bar{v}, \bar{w})$ in the network, we can reroute all flow that was traversing $\bar{a}$ from $\bar{v}$ to the sink $t$, while not changing any flow that was not affected by the interruption.
    Similar concepts were previously discussed in a few other papers \cite{brightwell2001reserving,chekuri2005building,phillips1999approximation,shepherd2009single}, but with an emphasis on network design issues, e.g., minimizing the cost of the installed capacity. In contrast, our interest is in computing maximum flows (but we point out that a potential application are feasibility/separation subroutines for capacity reservation). Note that in this setting, we cannot simply send a standard maximum flow, as we need to leave space for rerouting. Before we discuss our findings and better relate them to existing literature, let us formalize the definition of our model.

\paragraph{Network flows} Let $D = (V, A)$ be a digraph with source $s \in V$, a sink $t \in V$ and arc capacities $u \in \mathbb{R}_+^A$.
	Let $\paths \subseteq 2^A$ be the set of simple\footnote{All our results also work for the case that $\paths$ contains non-simple paths, but we restrict to simple paths for ease of notation.} $s$-$t$-paths in $D$.
For arcs $a, \bar{a} \in A$, define $$\paths[a] := \{P \in \paths \suchthat a \in P\} \quad \text{and} \quad \pathsVia{\bar{a}}{a} := \{P \in \paths \suchthat a, \bar{a} \in P, \, \bar{a} \prec_{P} a\},$$
where $\bar{a} \prec_{P} a$ means that $P$ traverses $\bar{a}$ before $a$. 
An \emph{$s$-$t$-flow} is a vector $x \in \mathbb{R}_+^{\paths}$ that assigns a flow value $x(P) \geq 0$ to each $P \in \paths$ such that the arc flow values $x(a) := \sum_{P \in \paths[a]} x(P)$ fulfill the capacity constraint $x(a) \leq u(a)$ for all~$a \in A$. The value of a flow $x$ is $\val{x} := \sum_{P \in \paths} x(P)$.

\paragraph{Reroutable flows}
Let $x$ be an $s$-$t$-flow. If an arc $\bar{a} = (\bar{v}, \bar{w}) \in A$ fails, all flow on paths containing the failing arc gets interrupted when it reaches $\bar{v}$.
For any $a \in A \setminus \{\bar{a}\}$, we define the \emph{available capacity} of $a$ after failure of $\bar{a}$ by
$$\rcap[\bar{a}]{x}{a} := u(a) - \elsum{P \in \paths[a] \setminus \pathsVia{\bar{a}}{a}} x(P).$$
A \emph{rerouting} of $x$ for the failing arc $\bar{a}$ is a $\bar{v}$-$t$-flow $x_{\bar{a}}$ of value $x(\bar{a})$ in $(V, A \setminus \{\bar{a}\})$ with capacities $\bar{u}_{x,\bar{a}}$. The flow $x$ is \emph{reroutable} if for every failing arc $\bar{a} \in A$ there is a rerouting $x_{\bar{a}}$ of $x$.
    
\paragraph{Strictly reroutable flows}
A rerouting $x_{\bar{a}}$ of a flow $x$ for a failing arc $\bar{a}$ is \emph{strict} if $x_{\bar{a}}(a) \leq \rcap{x}{a} := u(a) - x(a)$ for every $a \in A \setminus \{\bar{a}\}$. We say that $x$ is \emph{strictly reroutable} if for every failing arc $\bar{a} \in A$ there is a strict rerouting of $x$.
\bigskip

Strictly reroutable flows are both a helpful tool for computing reroutable flows and interesting in their own right, in situations where more conservative assumptions have to be made on the capacities available for rerouting.
 A natural question is what is the maximum flow value that can be sent by a (strictly) reroutable flow in a given network. We denote the corresponding optimization problem as {\maxreroute} and {\maxstrictreroute}, respectively.

\paragraph{Path flows vs.\ arc flows}
Our reroutable flow model is defined for path flows, which is a common assumption in many robust flow models~\cite{aggarwal2002multiroute,AnejaChandrasekaranNair2001,gottschalk2016robust,matuschke2016protecting}. It is well-known that network flows allow for two different representations: By specifying a value for each path, as is done in this paper, or by specifying the flow on each arc and requiring flow conservation at each node of the network. Arc flow values can easily be obtained from a given path flow. Conversely, every arc flow can be decomposed into flow on paths, but in general this decomposition is not unique; see, e.g.,~\cite{ahuja1993network}. There are flow problems where the path decomposition does not play a role, e.g., the maximum flow problem or the minimum cost flow problem, and others where it does, e.g., robust flows~\cite{AnejaChandrasekaranNair2001}, length-bounded flows~\cite{baier2006length}, or the computation of earliest arrival flows~\cite{minieka1973maximal}.
 For our model, it turns out that {\maxreroute} falls into the first category, whereas {\maxstrictreroute} falls into the second.

\subsection{Our results}

\paragraph{Complexity of the problems (\cref{sec:strict-lp,sec:strict-gap})}
We observe that {\maxstrictreroute} can be solved in polynomial time by formulating it as a linear program.
In contrast, {\maxreroute} is \NP-hard, even when $u(a) \in \{1, 2\}$ for all $a \in A$. On the positive side, by showing that the maximum value of a reroutable flow is at most twice as large as the maximum value of a strictly reroutable flow, we obtain a $2$-approximation for {\maxreroute} for arbitrary capacities. The problem can further be solved exactly in unit capacity networks (see below).

\paragraph{Max flow/min cut gap (\cref{sec:min-cut})} 
Max flow/min cut results play a central role in network flow theory.
We devise a combinatorial upper bound for the maximum reroutable flow value, called \emph{\rcut}, and prove that the corresponding flow/cut gap for both reroutable and strictly reroutable flows is bounded by $2$. In fact, our proof is constructive and provides a combinatorial $2$-approximation algorithm for the minimum capacity {\rcut} problem.

\paragraph{Unit capacity networks (\cref{sec:unit-capacities})}
We consider the case of unit capacities.
It turns out that in this case, {\maxreroute} and {\maxstrictreroute} are equivalent.
Our proof is based on a careful uncrossing argument that allows to transform any reroutable flow into a strictly reroutable flow.

\paragraph{Computing (half-)integral solutions (\cref{sec:integral-solutions})}
A common property of many flow problems is the existence of an integral optimal solution when capacities are integral. In the case of reroutable flows, this property does not hold. In fact, if we require flow to be integral, the problem becomes \NP-hard, even for sending a single unit of flow in a unit capacity network. 
However, for this special case, we devise a simple combinatorial algorithm that computes a half-integral solution or certifies that no flow of value $1$ exists. Via our max flow/min cut analysis we also show how to compute $2$-approximate half-integral solutions.

\paragraph{Multiple arc failures (\cref{sec:multiple-failures})}
We consider the natural generalization of our problems to multiple simultaneous arc-failures.
We show that in this case both variants of the problem are \NP-hard, even when only two arcs can fail and all arcs have unit capacity. All hardness results in this paper are based on reduction from an intermediary problem, called {\forbiddenpairs}. They are therefore grouped together in \cref{sec:hardness}.

\subsection{Related work}
As we already pointed out above, ``local'' rerouting schemes, i.e.,  schemes that only change flow affected by the failure, have been investigated in network design. A routing scheme in which flow has to be sent along arc-disjoint paths was investigated in \cite{brightwell2001reserving}, see also \cite{shepherd2009single}. The problem of finding a local rerouting from the tail to the head of a failed arc was investigated in \cite{chekuri2005building} and  \cite{phillips1999approximation}. However, in all these papers the focus was on min-cost capacity planning. 

Concepts that deal with the maximization of flow subject to robustness constraints commonly fall under the moniker of \emph{robust flows}. 
Aggarwal and Orlin~\cite{aggarwal2002multiroute} studied \emph{$k$-route flows}. Such a flow is a conic combination of elementary flows, each of which consists of a uniform flow along $k$ disjoint paths. Because of this structure, the failure of any arc can only destroy a $1/k$ fraction of the flow. A maximum $k$-route flow can be computed in polynomial time by means of a parametric max flow problem.
Another classic model is the \emph{maximum robust flow problem}: Here, the goal is to find a path flow that maximizes the surviving flow after a worst-case failure of $k$ arcs. Aneja et al.~\cite{AnejaChandrasekaranNair2001} showed that for $k=1$ both an optimal fractional and an optimal integral solution can be found in polynomial time.
If $k$ is not bounded by a constant the problem is \NP-hard~\cite{DisserMatuschke2016}, but the complexity for any constant value $k \geq 2$ is open.
Bertsimas et al.~\cite{bertsimas2013power} provide an $\Omega(1/k)$-approximation algorithm for the maximum robust flow. Robust flows are closely related to \emph{network flow interdiction}, which takes a dual perspective: The goal is to find a subset of arcs whose removal minimizes the maximum flow value in the remaining network; see the recent article by Chestnut and Zenklusen~\cite{chestnut2015hardness} for an up-to-date overview of this topic.

To the best of our knowledge, the only other flow maximization model that allows for adjustment after the failure are \emph{adaptive flows}, first introduced by Bertsimas et al.~\cite{BertsimasNasrabadiStiller2013}: In the first step, an arc flow is specified. After failure of $k$ arcs, a new flow is sent, with the flow value on every arc being bounded by the original flow value. 
Note that adaptive flows differ from reroutable flows in two important aspects:
Adaptive flows allow flow to be `lost'~(the flow value after the failure is lower than the original flow value), whereas in reroutable flows all flow has to reach the sink. Furthermore, adaptive flows can reconfigure the flow in the entire network, whereas in reroutable flows, only the flow affected by the failure can be~rerouted.

Another model closely related to reroutable flows is the \emph{online replacement path} problem~(ORP) introduced by Adjiashvili et al.~\cite{oriolo2013online}. The ORP is a generalization of the shortest path problem:
Given a digraph with costs on the arcs, we have to specify an $s$-$t$-path.
Along the path, we may encounter a failing arc $\bar{a} = \{\bar{v}, \bar{w}\}$, and we have to find a replacement path from $\bar{v}$ to $t$ avoiding $\bar{a}$. The goal is to minimize the total traveled distance, assuming $\bar{a}$ is chosen by an adversary.
Adjiashvili et al.~\cite{oriolo2013online} show that the ORP can be solved in polynomial time, even when a constant number of arcs fail.

\section{LP formulation, approximation, and\\ max flow/min cut}\label{sec:approximation}

In this section, we discuss the complexity of the two problems and provide bounds on the gap between {\maxreroute} and {\maxstrictreroute}. We also introduce an analogue to minimum cuts for reroutable flows and bound the corresponding duality gap. At the end of the section, we show that all our bounds are tight.

\subsection{Complexity of {\maxreroute} and {\maxstrictreroute}}\label{sec:strict-lp}

We now consider an LP formulation for {\maxstrictreroute}.
For $\bar{a} \in A$, let $\Rpaths{\bar{a}}$ be the set of all $\tail{\bar{a}}$-$t$-paths in $\NetRemoveArc{\bar{a}}$, which are exactly the paths that a rerouting for failing arc $\bar{a}$ can use.

\begin{alignat*}{3}
\LPstrict \qquad\qquad \mathllap{\max} \quad && \elsum{P \in \paths} x(P) & \\
\mathllap{\text{s.t.}} \quad && \elsum{P \in \paths[a]} x(P) \, + \, \elsum{\quad R \in \Rpaths{\bar{a}} \suchthat a \in R} x_{\bar{a}} (R) & \ \leq \ u(a) & \quad \forall\; a,\bar{a} \in A\\
&& \elsum{P \in \paths[\bar{a}]} x(P) \, - \, \elsum{R \in \Rpaths{\bar{a}}} x_{\bar{a}}(R) & \ = \ 0 & \quad \forall\; \bar{a} \in A\\
&& x, x_{\bar{a}} & \geq 0 & \quad \forall\; \bar{a} \in A
\end{alignat*}

The first set of constraints bound the capacities for each rerouting; note in particular that for~$\bar{a} = a$, the second term becomes $0$, ensuring $x(a) \leq u(a)$ for all $a \in A$. The second set of constraints ensures that the rerouting flow $x_{\bar{a}}$ has value $x(\bar{a})$.

For our discussion, it will also be useful to consider the dual of $\LPstrict$. We introduce dual variables $y_{\bar{a}}(a)$ for every $a, \bar{a} \in A$ and $z(\bar{a})$ for every $\bar{a} \in A$.

\begin{alignat*}{3}
\LPstrictD \qquad\qquad \mathllap{\min} \quad && \elsum{a \in A} \elsum{\bar{a} \in A} u(a) y_{\bar{a}}(a) & \\
\mathllap{\text{s.t.}} \quad && \elsum{a \in P} \Big(z(a) + \elsum{\bar{a} \in A} y_{\bar{a}}(a)\Big) & \ \geq \ 1 & \quad \forall\; P \in \paths\\
&& \elsum{a \in P} y_{\bar{a}}(a) \, - \, z(\bar{a}) & \ \geq \ 0 & \quad \forall\; \bar{a} \in A,\ P \in \Rpaths{\bar{a}} \\
&& y_{\bar{a}}(a) & \geq 0 & \quad \forall\; a, \bar{a} \in A
\end{alignat*}

Although $\LPstrict$ has an exponential number of variables, it can be solved in polynomial time via dual separation.
\begin{theorem}\label{thm:strict-poly}
  {\maxstrictreroute} can be solved in polynomial time.
\end{theorem}

\begin{proof}
We observe that the separation problem to $\LPstrictD$ can be solved by a sequence of shortest path computations: 
Consider a point $y, z$.
A violation of the first set of constraints can be detected by computing the shortest $s$-$t$-paths with respect to arc costs $c_{y,z}(a) := z(a) + \sum_{\bar{a} \in A} y_{\bar{a}}(a)$.
If the shortest path has cost smaller than $1$, it corresponds to a violated inequality, otherwise $y, z$ fulfills the inequalities.
 A violation of the second set of constraints can be detected by computing a shortest $\tail{\bar{a}}$-$t$-path for every $\bar{a} \in A$ with respect to arc costs $y_{\bar{a}}(a)$. If the shortest path for some $\bar{a}$ has cost smaller than $z(\bar{a})$ the constraint is violated for that path. Otherwise $y, z$ fulfills the inequalities.
 
 Hence we can obtain an optimal solution for $\LPstrictD$ using the equivalence of optimization and separation. By restricting $\LPstrict$ to the paths corresponding to inequalities generated during the solution of the dual, we obtain a primal solution of equal value.
\end{proof}

\begin{remark}
An alternative way to obtain a polynomial algorithm for {\maxstrictreroute} is to observe that the capacities $\rcap{x}{a} = u(a) - x(a)$ available for strict rerouting only depend on the arc flow values. Hence, we can formulate {\maxstrictreroute} as a compact LP with arc flow variables. Then any path decomposition  of the resulting arc flow is a maximum strictly reroutable flow.
\end{remark}

For the special case of unit capacity networks, we show in \cref{sec:unit-capacities} that an optimal solution to $\LPstrict$ is also optimal for {\maxreroute}.

\begin{theorem}
  For $u \equiv 1$, {\maxreroute} can be solved in polynomial time.
\end{theorem}

An LP formulation for {\maxreroute} can be obtained by replacing the term $\sum_{P \in \paths[a]} x(P)$ by $\sum_{P \in \pathsVia{\bar{a}}{a}} x(P)$ in the capacity constraints of $\LPstrict$.
Unfortunately, this modification prevents the dual separation approach from working. In fact, it turns out that {\maxreroute} is hard as soon as two different capacities occur. The proof of this result is discussed in \cref{sec:capacity-hardness}.

\begin{reptheorem}{thmGeneralCapacities}\label{thm:general-capacities-hardness}
  {\maxreroute} is \NP-hard, even when $u(a) \in \{1, 2\}$ for all $a \in A$.
\end{reptheorem}

\subsection{Reroutable flows vs.\ strictly reroutable flows}\label{sec:strict-gap}

As {\maxstrictreroute} is a tightening of {\maxreroute}, the optimal value of the former is at most that of the latter. We show that the gap between the two values cannot be larger than $2$. As we can compute maximum strictly reroutable flows, we obatain a $2$-approximation for {\maxreroute}.

\begin{lemma}\label{lem:2-bound}
Let $x$ be an $s$-$t$-flow. If $x$ is strictly reroutable, then $x$ is reroutable. If $x$ is reroutable, then $\frac{1}{2}x$ is strictly reroutable.
\end{lemma}

\begin{proof}
  The first statement follows from the fact that $\rcap[\bar{a}]{x}{a} \geq \rcap{x}{a}$ for all $\bar{a}, a \in A \setminus \{\bar{a}\}$. For the proof of the second statement, assume $x$ is reroutable and consider any $\bar{a} = (\bar{v}, \bar{w}) \in A$. Let $\bar{x}_{\bar{a}}$ be the rerouting of $x$ in case of failure of $\bar{a}$. Now observe that $\frac{1}{2}\bar{x}_{\bar{a}}$ is a strict rerouting of $\frac{1}{2}x$ in case of failure of $\bar{a}$, because it is a $\bar{v}$-$t$-flow of value $\frac{1}{2}x(\bar{a})$ and
  \begin{align*}
    \frac{1}{2}\bar{x}_{\bar{a}}(a) \ \leq \  
    \frac{1}{2}\left(u(a) - \elsum{P \in \paths[a] \setminus \pathsVia{\bar{a}}{a}} x(P)\right) \ \leq \ u(a) - \frac{1}{2}\elsum{P \in \paths[a]} x(P) \ = \ \bar{u}_{x/2}(a),
   \end{align*}
   where the second inequality follows from $\sum_{P \in \pathsVia{\bar{a}}{a}} x(P) \leq u(a)$.
\end{proof} 

\begin{corollary}
  There is a $2$-approximation algorithm for {\maxreroute}.
\end{corollary}

\subsection{Max flow/min cut gap for reroutable flows}\label{sec:min-cut}
An $s$-$t$-\emph{cut} is a set of arcs that intersects every $s$-$t$-path.
Its \emph{capacity} is the sum of capacities of its arcs.
A fundamental result in network flow theory is that the value of a maximum $s$-$t$-flow is equal to the capacity of a minimum $s$-$t$-cut.
This result has been successfully generalized to many variants of network flows, such as abstract flows~\cite{hoffman1974generalization} or flows over time~\cite{ford1962flows}. 
However, in other cases, such as multicommodity flows, the equality does not hold and instead, researchers investigate the worst case ratio between maximum flow and minimum cut; see, e.g., \cite{leighton1999multicommodity}.

We present a counterpart to an $s$-$t$-cut for reroutable flows.
It turns out that max flow and min cut are not necessarily equal and we give a tight bound on the corresponding max flow/min cut gap. 
An \emph{\rcut} is a set of arcs $R \subseteq A$ together with a collection of cuts $(C_a)_{a \in R}$, where each $C_a$ is a $\tail{a}$-$t$-cut containing $a$.
We denote $(R, (C_a)_{a \in R})$ by $(R, C)$ for short.
The capacity of the {\rcut} $(R, C)$ is 
$$\capacity{R, C} := \phi(R, C) + \sum_{a \in R} u(C_a \setminus \{a\}),$$
where $\phi(R, C)$ is the capacity of a minimum $s$-$t$-cut in $\NetRemoveSet{\cup_{a \in R} C_a}$.

The intuition behind this definition is the following: For every $a \in R$, all flow that crossed the cut $C_a$ must cross the $C_a \setminus \{a\}$ if $a$ fails. If a flow path does not cross any cut in $C_a$, then it crosses the minimum $s$-$t$-cut in $\NetRemoveSet{\cup_{a \in R} C_a}$. Therefore the capacity of an $R$-cut is an upper bound on the value of any reroutable flow. 

\begin{lemma}
  $\val{x} \leq \capacity{R, C}$ for any reroutable flow $x$ and any {\rcut} $(R,C)$.
\end{lemma}
\begin{proof}
  Let $a \in R$. As $C_{a}$ is a $\tail{a}$-$t$-cut, $x(a)$ units of flow have to be rerouted across the arcs in $C_{a} \setminus \{a\}$ when $a$ fails. Therefore
  $$x(a) \; \leq \; \elsum{a' \in C_{a} \setminus \{a\}} \rcap[a]{x}{a'} \; \leq \; u(C_{a} \setminus \{a\}) - \elsum{P \in \paths : P \cap C_{a} \neq \emptyset,\,a \notin P} x(P).$$
  This implies
  $\sum_{P \in \paths : P \cap C_{a} \neq \emptyset} x(P) \leq u(C_{a} \setminus \{a\})$ for all $a \in R$. Now let $S$ be a minimum $s$-$t$-cut in $\NetRemoveSet{\cup_{a \in R} C_a}$. Then 
  \begin{align*}
    \val{x} \; & = \; \elsum{P \in \paths} x(P) \; \leq \; \elsum{P \in \paths : P \cap S \neq \emptyset} x(P) + \elsum{a \in R} \elsum{\qquad P \in \paths : P \cap C_{a} \neq \emptyset} x(P) \\
    & \leq \; \sum_{a \in S} u(a) + \sum_{a \in R} u(C_a \setminus \{a\}) \; = \; \capacity{R, C}. \hspace{2.05cm}\qedhere
  \end{align*}  
\end{proof}

At the end of this section, we further show that {\rcut}s correspond to integral solutions to $\LPstrictD$. We now give a constructive proof bounding the duality gap between maximum strictly reroutable flow and minimum {\rcut} (or, equivalently, the integrality gap of the dual~LP). In \cref{sec:tightness} we give an example showing that the bound is tight.

\begin{theorem}\label{thm:cut-gap}
Let $x$ be a strictly reroutable flow of maximum value and let $(R, C)$ be an {\rcut} of minimum capacity. Then $\val{x} \geq \frac{1}{2} \capacity{R, C}$.
\end{theorem}

\begin{proof}
For $a \in A$, let $C_a$ be minimum $\tail{a}$-$t$-cut in $D$ containing $a$ and define $u'(a) := \min \{u(a), u(C_a \setminus\{a\})\}$. Let $C'$ be a minimum $s$-$t$-cut in $D$ with respect to the capacities $u'$ and let $x'$ be a corresponding maximum flow. Now define $R := \{a \in C' \suchthat u'(a) < u(a)\}$. Observe that $R$ and $(C_a)_{a \in R}$ define an {\rcut} and that $\phi(R, C) \leq u(C' \setminus R)$. We obtain $$\capacity{R, C} \leq \elsum{a \in C' \setminus R} u(a) + \sum_{a \in R} u(C_a \setminus \{a\}) = \sum_{a \in C'} u'(a) = \val{x'}.$$  Now let $x := x'/2$. It is sufficient to show that $x$ is a strictly reroutable flow.
By contradiction assume that there is $\bar{a} \in A$ for which there is no strict rerouting of $x$. By the max flow/min cut theorem, there must be a $\tail{\bar{a}}$-$t$-cut $\bar{C}$ in $\NetRemoveArc{\bar{a}}$ with $\sum_{a \in \bar{C}} \rcap{x}{a} < x(\bar{a})$. Note that $x(a) \leq u'(a)/2 \leq u(a)/2$ for every $a \in A$ by construction of $x$. Thus
$$\frac{1}{2} \sum_{a \in \bar{C}} u(a) \leq \sum_{a \in \bar{C}} (u(a) - x(a)) < x(\bar{a}) \leq \frac{1}{2}u'(\bar{a}) \leq \frac{1}{2}u(C_{\bar{a}} \setminus \{\bar{a}\}).$$
However, this implies that $\bar{C} \cup \{\bar{a}\}$ is a smaller $\tail{\bar{a}}$-$t$-cut than $C_{\bar{a}}$, a contradiction.
\end{proof}

\paragraph{Computing a minimum capacity {\rcut}}
Let us denote the problem of finding an {\rcut} of minimum capacity by {\minrcut}. The proof of \cref{thm:cut-gap} describes how to compute a $2$-approximate solution to this problem.

\begin{corollary}
  \mbox{There is a $2$-approximation algorithm for {\minrcut}.}
\end{corollary}
\begin{proof}
  For every $a \in A$ the capacity $u'(a)$ and the corresponding cut $C_a$ can be computed by a standard minimum $s$-$t$-cut computation. Given the values of $u'$, also $C'$ can be computed by a min cut computation.
\end{proof}

\subsubsection{{\rcut}s and integral dual solutions}
As mentioned above, {\rcut}s correspond to integral solutions of $\LPstrictD$. We give a formal argument of this correspondence.

Given an {\rcut} $(R, C)$, let $C^*$ be a minimum $s$-$t$-cut in $\NetRemoveSet{\bigcup_{a \in R}C_{a}}$.
We set $z(\bar{a}) := 1$ for every $\bar{a} \in R$, $y_{\bar{a}}(a) := 1$ for every $a \in C_{\bar{a}}$, and $y_{a}(a) := 1$ for every $a \in C^*$. All other variables are set to $0$. It is easy to check that $y, z$ corresponds to a feasible solution to $\LPstrictD$ with objective value $\capacity{R, C}$.

Conversely, consider an integral dual solution $y, z$. Note that we can assume that all variables take values in $\{0, 1\}$:
If $z(a) < 0$ we can set it to $0$ without losing feasibility. If any variable is takes a value larger than $1$, we can reduce it to $1$ without losing feasibility. 
Let $R := \{\bar{a} \in A \suchthat z(\bar{a}) = 1\}$.
By the second set of constraints, for every $\bar{a} \in R$, every $\tail{\bar{a}}$-$t$ path must be covered, i.e., there must be a $\tail{\bar{a}}$-$t$-cut $C_{\bar{a}}$ with $y_{\bar{a}}(a) = 1$ for all $a \in C_{\bar{a}} \setminus \{\bar{a}\}$.
By the first set of constraints, for every $P \in \paths$ there is an $a \in P$ such that either $z(a) = 1$ or $y_{\bar{a}}(a) = 1$ for some $\bar{a} \in A$.
Thus the support of $y$ contains an $s$-$t$-cut in $\NetRemoveSet{\bigcup_{a \in R}C_{a}}$. We conclude that $\capacity{R, C}$ is at most the objective value of the solution $y,z$.

\subsection{Summary of the bounds and tightness}\label{sec:tightness}
Putting the bounds from Lemma~\ref{lem:2-bound} and \cref{thm:cut-gap} together, we obtain the following corollary.

\begin{corollary}\label{cor:bounds}
  Let $(R, C)$ be a minimum capacity {\rcut} and let $x_{\textup{RF}}$ and $x_{\textup{SRF}}$ be maximal reroutable and strictly reroutable flows, respectively. Then\\[-10pt]
  $$\val{x_{\textup{RF}}} \leq \capacity{R, C} \leq 2\val{x_{\textup{SRF}}} \leq 2\val{x_{\textup{RF}}}.$$
\end{corollary}

\medskip

The following gadget will be useful throughout the paper in order to construct examples and reductions.

\paragraph{Backup links} A \emph{backup link} from $v$ to $w$ is a $v$-$w$-path $(a', a'')$ of length $2$ in which the intermediate node is incident only to the two arcs of the path and $u(a') := u(a'') := \max_{a \in A} u(a)$.
Note that $x(a') = x(a'') = 0$ for any reroutable flow, because when $a''$ fails, there is no $\tail{a''}$-$t$-path for rerouting the flow on that arc. A \emph{bidirected} backup link between $v$ and $w$ consists of two distinct backup links, one from $v$ to $w$ and one from $w$ to $v$.

\paragraph{Tightness of bounds}
The network depicted in \cref{fig:half-flow} shows that the bounds given in Lemma~\ref{lem:2-bound} and \cref{thm:cut-gap} are tight. Note that any reroutable flow is completely determined by the value $x(P)$ it sends along the path $P := \{a_1, a_2, a_3\}$, as backup links can only be used for rerouting.
  \begin{enumerate}
  \item For the bound on gap between max reroutable flow and max strictly reroutable flow, we set the capacities $u(a_1) = 2$ and $u(a) = 1$  for all $a \in A \setminus \{a_1\}$. We show that $x(P) = 1$ defines a reroutable flow. Failure of $a_2$ is not a problem, as $\tail{a_2}$ has a backup link to $t$. If $a_3$ fails, flow can use the backup link to $s$ and traverse~$a_1$, as $\rcap[a_3]{x}{a_1} = 2 - x(a_1) =1$, to reach the backup link from $\tail{a_2}$ to $t$. If $a_1$ fails, flow can be rerouted from $s$ using the backup link and the arc $a_3$, as $a_1 \prec_P a_3$ and thus $\rcap[a_1]{x}{a_3} = 1$.
  However, in a strictly reroutable flow, this last rerouting is no longer possible,  as $\rcap{x}{a_3} = 1 - x(a_3)$. Hence, $1 \geq x(a_1) + x(a_3) = 2x(P)$ for any strictly reroutable flow $x$. The maximum strictly reroutable flow value therefore is $1/2$.
  \item For the bound on the flow/cut gap, set all capacities to $1$.
  Then both the maximum reroutable flow value and the maximum strictly reroutable flow value are $1/2$. Consider any {\rcut} $(R, C)$. As capacities are integral, $\capacity{R,C} \geq 1$. An {\rcut} with $\capacity{R,C} = 1$ is, e.g., $R := \{a_3\}$ and $C_{a_3} := \{a_1, a_3\}$.
  \end{enumerate}
Note that \cref{fig:half-flow} also shows that optimal solutions to both {\maxreroute} and {\maxstrictreroute} can be fractional, even when capacities are integral.

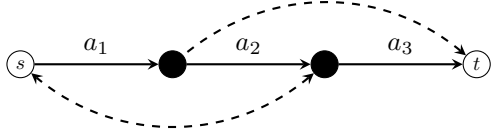
\begin{figure}[t]
\centering
\begin{tikzpicture}
  \path node[labeledNode] (a) {$s$}
  	++(2, 0) node[normalNode] (b) {}
  		edge[normalEdge, <-] node[above] {$a_1$} (a)
  	++(2, 0) node[normalNode] (c) {}
  		edge[normalEdge, <-] node[above] {$a_2$} (b)
  		edge[secondaryEdge, bend left=40, <->] (a)
  	++(2, 0) node[labeledNode] (w) {$t$}
  		edge[normalEdge, <-] node[above] {$a_3$} (c)
  		edge[secondaryEdge, bend right=40, <-] (b);
\end{tikzpicture}
\caption{Example showing that the bounds given in Lemma~\ref{lem:2-bound} and Theorem~\ref{thm:cut-gap} are tight. Dashed arcs correspond to (bidirected) backup links, which can only be used for rerouting. When all arcs have unit capacities, the maximum (strictly) reroutable flow has a value of $1/2$. When changing the capacity of $a_1$ to $2$, the maximum reroutable flow value increases to $1$, whereas the maximum strictly reroutable flow value remains $1/2$. The minimum {\rcut} capacity is $1$ in both cases.\label{fig:half-flow}}
\end{figure}

\begin{remark}
Note that the worst-case for the bounds in Corollary~\ref{cor:bounds} cannot be attained simultaneously, i.e., in any given instance either the max flow/min cut gap or the gap between reroutable and strictly reroutable flow has to be significantly smaller than $2$---in fact, at least one of them has to be within $\sqrt{2}$.
\end{remark}

\section{Unit capacity networks}\label{sec:unit-capacities}

Throughout this section, we assume $u \equiv 1$. We will show that in this case, any reroutable flow can be transformed into a strictly reroutable flow of the same value. While this closes the gap between the two reroutable flow variants, note that the flow/cut gap can still be $2$ in unit capacity networks, as can be seen in \cref{fig:half-flow}.
We start by giving an alternative characterization for strictly reroutable flows in unit capacity networks.
\vspace{-0.2cm}

\paragraph{Cuts separating $t$} For $S \subseteq V$, let $\outarcs{S} := \{a \in A \,:\, \tail{a} \in S,\, \head{a} \in V \setminus S\}$ denote the cut induced by $S$. We define $\mathcal{S} := \{S \subset V \setminus \{t\} \,:\, S \neq \emptyset\}$ and let $\mathcal{C} := \{\outarcs{S} \,:\, S \in \mathcal{S}\}$ be the set of \emph{$t$-separating cuts}.
 W.l.o.g.~we assume $\outarcs{S} \neq \emptyset$ for all $S \in \mathcal{S}$, as no vertex in a set $S$ with $\outarcs{S} = \emptyset$ can be on an $s$-$t$-path.
 
\begin{lemma}\label{lem:spanning-tree-characterization}
  Let $x$ be an $s$-$t$-flow for capacities $u \equiv 1$.
  Then $x$ is strictly reroutable if and only if $\sum_{a \in C} (1 - x(a)) \geq 1$ for all $C \in \mathcal{C}$.
\end{lemma}
\begin{proof}
  We first show sufficiency of the condition. By contradiction assume that \mbox{$\sum_{a \in C} (1 - x(a)) \geq 1$} for all $C \in \mathcal{C}$ but $x$ is not strictly reroutable. Because $x$ is not strictly reroutable, there must be an arc $\bar{a} \in A$ such that there is no rerouting of $x$ for $\bar{a}$. This means that the maximum flow value that can be sent in $\NetRemoveArc{\bar{a}}$ with capacities $\rcap{x}{}$ from $\bar{v} := \tail{\bar{a}}$ to $t$ is strictly smaller than $x(\bar{a})$. By max flow/min cut, this implies there is a $\bar{v}$-$t$-cut $C \in \mathcal{C}$ with $\sum_{a \in C \setminus \{\bar{a}\}} \rcap{x}{a} < x(\bar{a})$, which implies $\sum_{a \in C} 1 - x(a) < 1$, contradicting our initial assumption.
  
  To see necessity, assume $x$ is a strictly reroutable flow and let $C \in \mathcal{C}$ be any $t$-separating cut. Now let $\bar{a} \in C$. Since $x$ is strictly reroutable, there is an $\tail{\bar{a}}$-$t$-flow of value $x(\bar{a})$ in $\NetRemoveArc{\bar{a}}$ with capacites $\rcap{x}{}$. By max flow/min cut this implies $\sum_{a \in C \setminus \{\bar{a}\}} \rcap{x}{a} \geq x(\bar{a})$.
\end{proof}

In the following, we identify those cuts that might violate the condition given in Lemma~\ref{lem:spanning-tree-characterization} for a (non-strictly) reroutable flow. We then show that this class of cuts forms a semi-lattice. This allows us to apply an uncrossing of the flow paths that iteratively eliminates the problematic cuts while maintaining reroutability.
\vspace{-0.2cm}

\paragraph{Bad cuts}
Let $x$ be an $s$-$t$-flow and let $C \in \mathcal{C}$ be a $t$-separating cut. An arc $\bar{a} \in C$ is \emph{$(x, C)$-bad} if there is an arc $a \in C$ and a path $P \in \pathsVia{\bar{a}}{a}$ with $x(P) > 0$. A cut $C$ is \emph{$x$-bad} if all arcs $\bar{a} \in C$ are $(x, C)$-bad.

\begin{lemma}\label{lem:good-cuts}
  Let $x$ be a reroutable flow for capacities $u \equiv 1$. Let $C \in \mathcal{C}$ be a $t$-separating cut. 
  If $\sum_{a \in C} (1 - x(a)) < 1$ then $C$ is $x$-bad.
\end{lemma}

\begin{proof}
  By contradiction assume $C$ is not $x$-bad. Then there must be an arc $\bar{a} \in C$ that is not $(x, C)$-bad. This implies that $\sum_{P \in \pathsVia{\bar{a}}{a}} x(P) = 0$ for every $a \in C \setminus \{\bar{a}\}$. In particular, $\rcap[\bar{a}]{x}{a} = \rcap{x}{a} = 1 - x(a)$ for all $a \in C \setminus \{\bar{a}\}$. Since all flow in the rerouting of $x$ for failure of $\bar{a}$ needs to cross $C \setminus \{\bar{a}\}$, we obtain $\sum_{a \in C \setminus \{\bar{a}\}} \rcap[\bar{a}]{x}{a} \geq x(\bar{a})$. Adding $1 - x(\bar{a})$ to both sides of this inequality yields a contradiction.
\end{proof}

\begin{lemma}\label{lem:rightmost-cut}
  Let $x$ be a flow and let $S, S' \in \mathcal{S}$ be such that $\outarcs{S}$ and $\outarcs{S'}$ are both $x$-bad. Then $\outarcs{S \cup S'}$ is an $x$-bad $t$-separating cut as well.
\end{lemma}
\begin{proof}
  Define $C := \outarcs{S}$, $C' := \outarcs{S'}$, and $C^* := \outarcs{S \cup S'}$. Let $\bar{a} \in C^*$. We will show that $\bar{a}$ is $(x, C^*)$-bad, proving the lemma. 
  Note that $C^* \subseteq C \cup C'$ and hence $\bar{a} \in C$ or $\bar{a} \in C'$. Without loss of generality assume the former. Because, $C$ is $x$-bad, $\bar{a}$ is $(x, C)$-bad.
  Therefore there must be an arc $a \in C$ and a path $P \in \pathsVia{\bar{a}}{a}$ with $x(P) > 0$. Consider the suffix $Q := P[\tail{a}, t]$ of $P$ starting with arc $a$. Observe that $Q$ starts in $S \cup S'$ but ends in $t \notin S \cup S'$.  Hence $Q$ crosses $C^*$, i.e., there is $a' \in Q \cap C^* \subseteq P \cap C^*$. Observe that $\bar{a} \prec_{P} a \preceq_P a'$, i.e., $P \in \pathsVia{\bar{a}}{a'}$, showing that $\bar{a}$ is $(x, C^*)$-bad.
\end{proof}

\paragraph{Uncrossing paths}
	 Let $P \in \paths$. For two nodes $v, w \in V$ visited by $P$ (in that order), we let $P[v, w]$ denote the subpath of path $P$ starting at $v$ and ending at $w$.
	 Given another path $Q \in \mathcal{P}$ and an arc $a \in P \cap Q$, let $P \times_a Q$ be a simple $s$-$t$-path in the concatenation of $P[s, \head{a}]$ and $Q[\head{a}, t]$.

\begin{theorem}\label{thm:equivalence}
  Let $x$ be a reroutable flow for capacities $u \equiv 1$. Then there is a strictly reroutable flow $x'$ with $\val{x'} = \val{x}$ and $x'(a) \leq x(a)$ for all $a \in A$.
\end{theorem}

\begin{proof}
  W.l.o.g., assume that $x$ minimizes $\sum_{a \in A} x(a)$ among all reroutable flows $x'$ with $\val{x'} = \val{x}$ and $x'(a) \leq x(a)$ for all $a \in A$ (if this is not the case, we can replace $x$ by the flow minimizing the total arc flow).

  Define $\mathcal{S}' := \{S \in \mathcal{S} \,:\, \outarcs{S} \text{ is an $x$-bad cut}\}$. If $\mathcal{S}' = \emptyset$, 
 Lemmas~\ref{lem:spanning-tree-characterization} and~\ref{lem:good-cuts} imply that $x$ is strictly reroutable and we are done.
  Thus assume $\mathcal{S}' \neq \emptyset$ and define $S^* := \bigcup_{S \in \mathcal{S}'} S$ and $C^* := \outarcs{S^*}$. Note that Lemma~\ref{lem:rightmost-cut} implies $S^* \in \mathcal{S}'$, and that further, by construction, $S^*$ defines the rightmost bad cut, i.e., no cut $\outarcs{S}$ for $S \in \mathcal{S}$ with $S \setminus S^* \neq \emptyset$ is $x$-bad. 
  
  Next we construct a digraph $H = (V_H, A_H)$ as follows. We let $V_H := C^*$, i.e., the nodes of $H$ are the arcs of $C^*$. For every pair of distinct arcs $a, a' \in C^*$, we introduce the arc $(a, a')$ in $A_H$ if and only if there is a path $P \in \supp{x} \cap \pathsVia{a}{a'}$ such that $a'$ is the last arc of $C^*$ on $P$. Observe that, because $C^*$ is $x$-bad, every node of $H$ has an outgoing arc. Hence $H$ contains a simple directed cycle $Z$.
 Let $a_1, \dots, a_k \in C^*$ be the arcs corresponding to the nodes of the cycle $Z$, and let $P_1, \dots, P_k \in \supp{x}$ be paths corresponding to the arcs of $Z$, i.e., the paths fulfill that arc $a_i$ is the last arc of $P_i$ that crosses $C^*$, and $a_i \in P_i \cap P_{i+1}$ for each $i \in \{1, \dots, k\}$ (for ease of notation we identify $i$ and $j$ if $i \equiv j \mod k$). 
 
  Now define $P'_i := P_{i+1} \times_{a_i} P_i$ for $i \in [k]$. Let $\varepsilon := \min_i x(P_i)$; see \cref{fig:uncrossing} for an illustration. We construct a new flow $x'$ as follows:
  \begin{align*}
    x'(P) = \begin{cases}
      x(P) + \varepsilon & \text{if } P = P'_i \text{ for some } i,\\
      x(P) - \varepsilon & \text{if } P = P_i \text{ for some } i,\\
      x(P) & \text{otherwise}.
    \end{cases}
  \end{align*}
  We show that $x'$ is also a reroutable flow. First observe that $\val{x'} = \val{x}$ and that $x'(a) \leq x(a)$ for all $a \in A$.
  Now let $\bar{a} \in A$ and let $S \subseteq V \setminus \{t\}$ with $\tail{\bar{a}} \in S$  and define $C := \outarcs{S}$. First observe that if $S \setminus S^* \neq \emptyset$, then $C$ is not $x$-bad and therefore $\sum_{a \in C} (1 - x'(a)) \geq \sum_{a \in C} (1 - x(a)) \geq 1$ by Lemma~\ref{lem:good-cuts}, implying $\sum_{a \in C \setminus \{\bar{a}\}} \rcap[\bar{a}]{x'}{a} \geq x'(\bar{a})$, and therefore $x'$ is reroutable. Thus we consider the case $S \subseteq S^*$.
  We will show that in this case $\rcap[\bar{a}]{x'}{a} \geq \rcap[\bar{a}]{x}{a}$ for all $a \in C$, and therefore $x'$ is again reroutable. To this end, observe that the definition of $\rcap[\bar{a}]{x}{}$ implies
  \begin{align*}
    \rcapEq[\bar{a}]{x'}{a} - \rcapEq[\bar{a}]{x}{a} & = \ \ 
    \elsum{P \in \paths[a] \setminus \pathsVia{\bar{a}}{a}} x(P) \ - \ \ \elsum{P \in \paths[a] \setminus \pathsVia{\bar{a}}{a}} x'(P)\\
    & = \varepsilon \cdot \big(|\underbrace{\{i \,:\, P_i \in \paths[a] \setminus \pathsVia{\bar{a}}{a}\}}_{=:I}| - |\underbrace{\{i \,:\, P'_i \in \paths[a] \setminus \pathsVia{\bar{a}}{a}\}}_{=:I'}| \big).
  \end{align*}
  We show that $i \in I'$ implies $i + 1 \in I$ and therefore $|I| \geq |I'|$, which proves our claim. Consider any $i \in I'$. We observe that $a \in P_{i+1}[s, \head{a_i}]$, as $\tail{a} \in S \subseteq S^*$ and no arc in $P_{i}[\head{a_i}, t]$ has its tail in $S^*$ (recall that $a_i$ is the last arc of $P_i$ crossing $C^*$). This further implies that $P_{i+1}[s,\head{a}] \subseteq P'_i[s,\head{a}]$ and thus $\bar{a} \notin P_{i+1}[s,\head{a}]$ as $P'_i \notin \pathsVia{\bar{a}}{a}$. Therefore $P_{i+1} \in \paths[a] \setminus \pathsVia{\bar{a}}{a}$, i.e., $i+1 \in I$. Now $|I| \geq |I'|$ implies $\rcap[\bar{a}]{x'}{a} \geq \rcap[\bar{a}]{x}{a}$ and hence $x'$ is reroutable.
    
  Finally, we show that $\sum_{a \in A} x'(a) < \sum_{a \in A} x(a)$. To this end, consider any fixed $i \in [k]$. Observe that $$x(a_i) - x'(a_i) = \varepsilon \cdot \left( |\{j \suchthat a_i \in P_j\}| - |\{j \suchthat a_i \in P'_j\}| \right).$$
  Note that $P'_{j} \cap C^* \subseteq P_{j+1}$ and hence $a_i \in P'_j$ implies $a_i \in P_{j+1}$, i.e., $\{j \suchthat a_i \in P'_j\} \subseteq \{j - 1 \suchthat a_i \in P_j\}$. 
  Further note that $P'_{i - 1} = P_{i} \times_{a_{i-1}} P_{i-1}$ does not contain $a_i$, as $a_{i-1} \prec_{P_i} a_i$.
  Hence the above containment is strict and $x(a_i) - x'(a_i) > 0$.
  
  We thus have constructed a reroutable flow $x'$ with $\val{x'} = \val{x}$ and $x'(a) \leq x(a)$ for all $a \in A$ and $\sum_{a \in A} x'(a) < \sum_{a \in A} x(a)$, contradicting our initial assumption.
\end{proof}
 
\begin{figure}[t]
\hspace*{-0.7cm}%
  \begin{tikzpicture}[xscale=1.1]
    \tikzstyle{path1}=[normalEdge, solid, red]
    \tikzstyle{path2}=[normalEdge, dotted, very thick, blue]
    \tikzstyle{path3}=[normalEdge, loosely dashed]
  
    \path node[labeledNode] (s) {$s$}
    ++(2, 1) node[normalNode] (v1) {}
      edge[path1, <-] node[sloped, above] {$P_1$}  (s)
    +(2, 0) node[normalNode] (w1) {} 
      edge[path1, <-] (v1)
      edge[path3, <-] node[above] {\textcolor{black}{$a_3$}} (v1)
    ++(0, -1) node[normalNode] (v2) {}
      edge[path1, <-] (w1)
      edge[path2, <-] node[sloped, above] {$\quad P_2$} (s)
    +(2, 0) node[normalNode] (w2) {} 
      edge[path1, <-] node[above] {\textcolor{black}{$\qquad a_1$}} (v2)
      edge[path2, <-] (v2)
    ++(0, -1) node[normalNode] (v3) {}
      edge[path2, <-] (w2)
      edge[path3, <-] node[sloped, below] {$P_3$} (s)
    +(2, 0) node[normalNode] (w3) {} 
      edge[path2, <-] node[below] {\textcolor{black}{$a_2$}} (v3)
      edge[path3, <-] (v3)
      edge[path3, ->] (v1)
    ++(4, 1) node[labeledNode] (t) {$t$}
      edge[path2, <-] (w3)
      edge[path1, <-] (w2)
      edge[path3, <-] (w1);
  
    \draw (6.5, 1.5) -- ++(0, -3);
  
    \path (7, 0) node[labeledNode] (s) {$s$}
    ++(2, 1) node[normalNode] (v1) {}
      edge[path1, <-] node[sloped, above] {$P'_3$}  (s)
    +(2, 0) node[normalNode] (w1) {} 
      edge[path3, <-] node[above] {\textcolor{black}{$a_3$}} (v1)
    ++(0, -1) node[normalNode] (v2) {}
      edge[path2, <-] node[sloped, above] {$\quad P'_1$} (s)
    +(2, 0) node[normalNode] (w2) {} 
      edge[path1, <-] node[above] {\textcolor{black}{$a_1$}} (v2)
    ++(0, -1) node[normalNode] (v3) {}
      edge[path3, <-] node[sloped, below] {$P'_2$} (s)
    +(2, 0) node[normalNode] (w3) {} 
      edge[path2, <-] node[above] {\textcolor{black}{$a_2$}} (v3)
    ++(4, 1) node[labeledNode] (t) {$t$}
      edge[path2, <-] (w3)
      edge[path1, <-] (w2)
      edge[path3, <-] (w1);
    
  \end{tikzpicture}
  \caption{Uncrossing of paths on a bad cut.}
  \label{fig:uncrossing}
\end{figure}
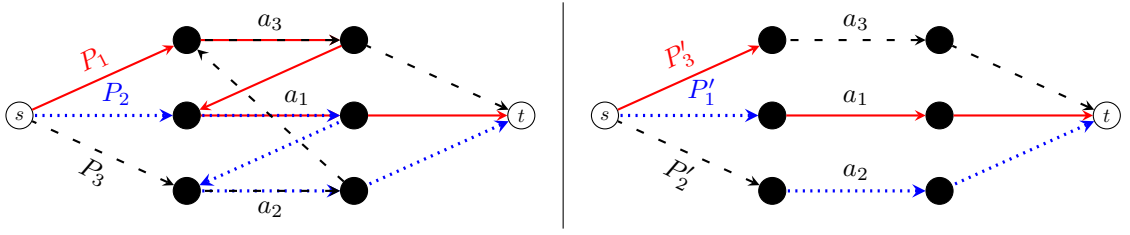

\begin{remark}\label{rem:integral-equivalence}
  The proof of \cref{thm:equivalence} preserves integrality. Therefore, if $x(P)$ is an integer multiple of $\alpha$ for every $P \in \mathcal{P}$, then $x'$ can be chosen such that also $x'(P)$ is an integer multiple of $\alpha$ for every $P \in \mathcal{P}$.
\end{remark}

\begin{remark}
  The characterization of strictly reroutable flows for unit capacities given in Lemma~\ref{lem:spanning-tree-characterization} can be extended to instances with arbitrary capacities as follows: An $s$-$t$-flow $x$ is strictly reroutable if and only if $\sum_{a \in C \setminus \{\bar{a}\}} (u(a) - x(a)) \geq x(\bar{a})$ for all $C \in \mathcal{C}$ and all $\bar{a} \in C$. 
  Furhtermore, if flow $x$ violates this constraint for some cut $C$ and an arc $\bar{a} \in C$, then $\bar{a}$ is $(x, C)$-bad. However, this no longer implies that $C$ is an $x$-bad cut. Indeed, consider the example given in \cref{fig:half-flow} when setting $u(a_1) = 2$ and $u(a) = 1$  for all $a \in A \setminus \{a_1\}$. 
Let $x$ be the reroutable (but not strictly reroutable) flow sending one unit of flow along the path $(a_1, a_2, a_3)$. The cut $C := \{a_1, a_3\}$ is not $x$-bad, but $C$ and $a_1$ violate the above constraint, preventing a strict rerouting.  
\end{remark}

\section{Computing (half-)integral solutions}\label{sec:integral-solutions}

In some application contexts, flow cannot be split into arbitrarily small pieces. This is the setting we consider in this section. We say a flow $x$ is \emph{integral}, if $x(P) \in \mathbb{Z}$ for all $P \in \paths$. We say that $x$ is \emph{\halfint} if $2x$ is integral.

For many fundamental flow problems, such as \textsc{Max Flow} or \textsc{Min Cost Flow}, integrality comes for free, i.e., as long as capacities are integral, there exists an optimal integral solution. In the case of reroutable flows, this property does not hold,
see, e.g., \cref{fig:fractionality}.
In fact, it turns out to be \NP-hard to decide whether there is a non-zero integral reroutable flow in a network.

\begin{figure}[t]
\hspace*{-0.6cm}%
\begin{tikzpicture}[scale=0.9]
  \node[labeledNode] (s) {$s$};
  \path (s)
  	++(1.5, 1.5) node[normalNode] (u11) {}
  		edge[normalEdge, <-] (s)
  	++(2, 0) node[normalNode] (u12) {}
  		edge[normalEdge, <-] (u11)
  	++(2, 0) node[normalNode] (u13) {}
  		edge[normalEdge, <-] (u12);
  \path (s)
  	++(1.5, 0.25) node[normalNode] (u21) {}
  		edge[normalEdge, <-] (s)
  	++(2, 0) node[normalNode] (u22) {}
  		edge[normalEdge, <-] (u21)
  	++(2, 0) node[normalNode] (u23) {}
  		edge[normalEdge, <-] (u22);
  \path (s)
  	++(1.5, -1.5) node[normalNode] (u31) {}
  		edge[normalEdge, <-] (s)
  	++(2, 0) node[normalNode] (u32) {}
  		edge[normalEdge, <-] (u31)
  	++(2, 0) node[normalNode] (u33) {}
  		edge[normalEdge, <-] (u32)
  	++(2, 1.5) node[labeledNode] (v) {$v$}
  		edge[normalEdge, bend left=15, <-] (s)
  		edge[normalEdge, <-] (u13)
  		edge[normalEdge, <-] (u23)
  		edge[normalEdge, <-] (u33)
  	++(2, 1.5) node[normalNode] (v11) {}
  		edge[normalEdge, <-] (v)
  		edge[secondaryEdge, bend right=40, <->] (u11)
  	++(2, 0) node[normalNode] (v12) {}
  		edge[normalEdge, <-] (v11)
  		edge[secondaryEdge, <->, bend right=40] (u21)
  	++(2, 0) node[normalNode] (v13) {}
  		edge[normalEdge, <-] (v12)
  		edge[secondaryEdge, bend right=42, <->] (u31);
   \path (v)
  	++(2, 0) node[normalNode] (v21) {}
  		edge[normalEdge, <-] (v)
  		edge[secondaryEdge, bend right=45, <->] (u12)
  	++(2, 0) node[normalNode] (v22) {}
  		edge[normalEdge, <-] (v21)
  		edge[secondaryEdge, bend right=20, <->] (u22)
  	++(2, 0) node[normalNode] (v23) {}
  		edge[normalEdge, <-] (v22)
  		edge[secondaryEdge, <->, bend right=25] (u32);
   \path (v)
  	++(2, -1.5) node[normalNode] (v31) {}
  		edge[normalEdge, <-] (v)
  		edge[secondaryEdge, bend left=30, <->] (u13)
  	++(2, 0) node[normalNode] (v32) {}
  		edge[normalEdge, <-] (v31)
  		edge[secondaryEdge, bend left=40, <->] (u23)
  	++(2, 0) node[normalNode] (v33) {}
  		edge[normalEdge, <-] (v32)
  		edge[secondaryEdge, bend left=30, <->] (u33)
  	++(1.5, 1.5) node[labeledNode] (t) {$t$}
  	  edge[normalEdge, <-] (v13)
  	  edge[normalEdge, <-] (v23)
  	  edge[normalEdge, <-] (v33);
\end{tikzpicture}

\caption{Example network in which no integral or half-integral reroutable flow is optimal. Dashed arcs represent bidirected backup links (see \cref{sec:tightness}), all arcs have unit capacities. The maximum reroutable flow value is $2$. This can only be achieved when $x(s, v) = 1$, the three $s$-$v$-paths all carry $1/3$ unit of flow, and the three $v$-$t$-paths all carry $2/3$ unit of flow. See Remark~\ref{rem:non-half-int-example} for a detailed discussion.}\label{fig:fractionality}
\end{figure}
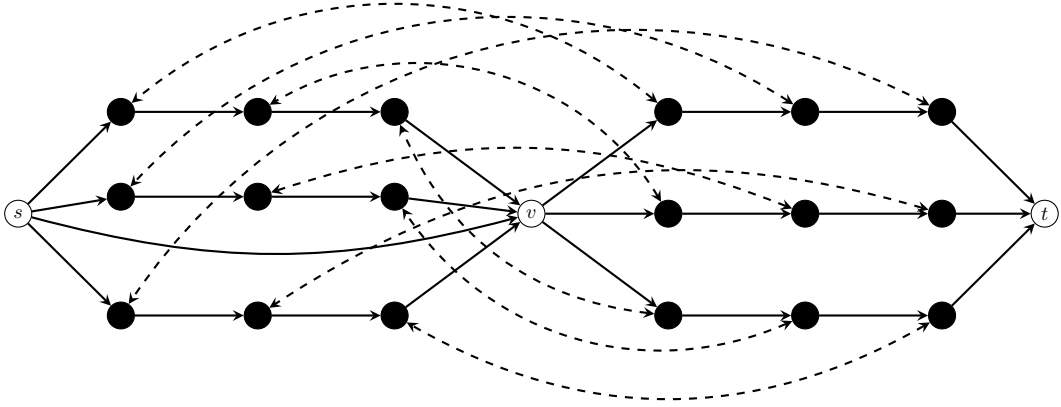

\begin{reptheorem}{thmIntegralHardness}\label{thm:integral-hard}
  It is \NP-hard to decide whether there is an integral (strictly) reroutable flow of value $1$, even when restricted to instances with $u \equiv 1$.
\end{reptheorem}

Note that this problem corresponds to sending a unit of flow along a single $s$-$t$-path. The hardness stems from a problem named {\forbiddenpairs}. The reduction is described in \cref{sec:hardness}.
While it seems that \cref{thm:integral-hard} does not give much space for positive algorithmic results, we can do much better if we relax the integrality requirement slightly. 

\begin{reptheorem}{thmCombAlgo}\label{thm:comb-alg}
Given a network with $u \equiv 1$, the algorithm given in \cref{alg:unit-demand} computes in polynomial time either a half-integral strictly reroutable flow of value $1$, or correctly determines that no reroutable flow of value $1$ exists.
\end{reptheorem}

In particular, this implies that if we are interested in sending a single unit of flow, we never need to split our flow in more than two paths.
Before we discuss the algorithm from \cref{thm:comb-alg}, let us shortly discuss the case of arbitrary capacities. As a consequence of the max flow/min cut result proven in \cref{sec:min-cut}, we obtain the following approximation.

\begin{theorem}\label{thm:half-approx}
If $u$ is integral, then there is a strictly reroutable half-integral flow $x$ with $\val{x} \geq \operatorname{OPT} / 2$, where $\operatorname{OPT}$ is the value of a maximum reroutable flow. The flow $x$ can be computed in polynomial time.
\end{theorem}

\begin{proof}
Recall that in the proof of \cref{thm:cut-gap} we computed an $s$-$t$-flow $x'$ that was maximal with respect to capacities $u'(a) := \min \{u(a), u(C_{a} \setminus \{a\})\}$. We then showed that the flow~$x := x'/2$ is strictly reroutable and within a factor of $2$ of a corresponding {\rcut}. In particular, $\val{x}$ is within a factor of $2$ of the maximum reroutable flow value.
Note that if~$u$ is integral, also $u'$ is integral, and hence we can choose $x'$ to be integral, ensuring that $x$ is half-integral.
\end{proof}

\begin{lstlisting}[caption={Computing a half-integral reroutable unit demand flow},label=alg:unit-demand,captionpos=b,float,escapeinside={(*}{*)},backgroundcolor=\color{lightgray},rulecolor=\color{lightgray}]
(*$A_0 := \emptyset, A_1 := \emptyset$*)
while (*$\exists\, a \in A \setminus A_0 \,:\, A_1 \cup \{a\}$ is a $\tail{a}$-$t$-cut in $D$*)
  (*$A_0 := A_0 \cup \{a\}$*)
  (*$A_1 \leftarrow \{a' : a' \text{ is an $s$-$t$-bridge in } \NetRemoveSet{A_0}\}$*)
end while
if (*$A_0$ is an $s$-$t$-cut in $D$*)
  return (*''No reroutable flow of value $1$ exists.``*)
else
  (*Let $P_1, P_2$ be two $s$-$t$-paths in $\NetRemoveSet{A_0}$ such that $P_1 \cap P_2 = A_1$.*)
  (*Let $x$ be the flow defined by $x(P_1) = x(P_2) = 1/2$.*)
  return (*$x$*)
end if
\end{lstlisting}

\paragraph{Algorithm for computing a half-integral flow for unit demand}
A natural starting point for an algorithm is to identify arcs $a \in A$ such that $\tail{a}$ is disconnected from $t$ in $\NetRemoveArc{a}$. Obviously, no reroutable flow can send a positive amount of flow along such arcs, as after failure of $a$, the flow cannot be rerouted to $t$. Surprisingly, this simple preprocessing step can be generalized to an iterative procedure that solves the problem.

The algorithm, which is formally given in Listing~\ref{alg:unit-demand}, maintains two sets $A_0$ and $A_1$. In every iteration, it identifies an arc that cannot carry any flow in any reroutable flow and adds it to $A_0$. The set $A_1$ contains the $s$-$t$-bridges in the graph $\NetRemoveSet{A_0}$, i.e., all arcs whose removal disconnects $s$ from $t$ in that graph. Clearly, if $x(a) = 0$ for all $a \in A_0$, then every arc in $A_1$ must carry $1$ unit of flow.
If at some point $A_0$ becomes an $s$-$t$-cut, we know that no reroutable flow of value~$1$ exists. On the other hand, if the algorithm finds no more arcs to add to $A_0$ while $s$ and $t$ are still connected in $\NetRemoveSet{A_0}$, it computes two paths $P_1, P_2$ that only intersect at the bridges, and sends $1/2$ units of flow along each of them. 

\begin{proof}[Proof of \cref{thmCombAlgo}]
  To see that Algorithm~\ref{alg:unit-demand} terminates in polynomial time, observe that $|A_0|$ is increased in every iteration of the while-loop and the loop thus terminates after at most $|A|$ iterations, each of which can be carried out in polynomial time.
  
\textit{Case 1: No flow exists.} 
We now show that if Algorithm~\ref{alg:unit-demand} denies the existence of a reroutable flow of value $1$, this is indeed correct. By contradiction assume $A_0$ contains an $s$-$t$-cut but there exists a reroutable flow $x$ of value $1$. We prove by induction that at any step of algorithm the set $A_0$ fulfills the property that $x(a) = 0$ for all $a \in A_0$, yielding a contradiction. The claim is clearly true initially, when $A_0 = \emptyset$.
Now consider any iteration of the while-loop, considering arc $a$. By induction hypothesis, every $s$-$t$-path $P$ with $x(P) > 0$ must be a path in $\NetRemoveSet{A_0}$.
Note that there is an order $a_1, \dots, a_{\ell}$ of the set $A_1$ of $s$-$t$-bridges of $\NetRemoveSet{A_0}$ such that every such flow-carrying path contains all of these bridges in exactly that order. In particular $x(a_1) = \ldots = x(a_{\ell}) = 1$. 
Now consider the next arc $a$ added to $A_0$ and assume by contradiction that $x(a) > 0$. By choice of $a$ there is a $\tail{a}$-$t$-cut $C \subseteq A_1 \cup \{a\}$ in $D$. 
Note that if $C \cap A_1 = \emptyset$, there is no rerouting of $x$ in case of failure of arc $a$, as there is no $\tail{a}$-$t$-path in $\NetRemoveArc{a}$. Thus, let $a_k \in C \cap A_1$ be the bridge with the highest index $k$ on the cut. 
We distinguish two cases: 
\begin{itemize}
\item[(i)] Assume $a$ appears before $a_k$ on every flow-carrying path.
Note that $C$ is a $\tail{a_k}$-$t$-cut because $a_k \in C$ and that $\sum_{a' \in C} \rcap[a_k]{x}{a'} = 1 - x(a) < 1$. Therefore, the one unit of flow on $a_k$ cannot be rerouted when $a_k$ fails.
\item[(ii)] Now assume $a$ occurs after $a_k$ on every flow-carrying path. But then, when $a$ fails, the flow on $a$ cannot be rerouted as all edges in $C \setminus \{a\} \subseteq A_1$ occur before $a$ on every flow-carrying path and thus $\sum_{a' \in C \setminus \{a\}} \rcap[a]{x}{a'} = 0$. 
\end{itemize}
We thus deduce that $x(a) = 0$, completing the induction.

\textit{Case 2: Algorithm returns flow.}
Finally, we show that if $\NetRemoveSet{A_0}$ contains an $s$-$t$-path after completing the while-loop, then the flow $x$ returned by the algorithm is a strictly reroutable flow. First observe that two $s$-$t$-paths $P_1, P_2$ in $\NetRemoveSet{A_0}$ with $P_1 \cap P_2 = A_1$ exist by the max flow/min cut theorem, as $A_1$ contains exactly the bridges of $\NetRemoveSet{A_0}$. Now consider the failure of any arc $\bar{a} \in A \setminus A_0$. Let $C$ be a $\tail{\bar{a}}$-$t$-cut in $D$ minimizing $U(C) := \sum_{a \in C \setminus \{\bar{a}\}} \rcap{x}{a}$.
We show that $U(C) \geq x(\bar{a})$, which by max flow/min cut implies that there is a rerouting of $x$ in case of failure of $\bar{a}$.
 By termination condition of the while-loop, there is at least one arc $a' \in C \setminus (A_1 \cup \{\bar{a}\})$. Note that $x(a') \in \{0,\, 1/2\}$ and thus $U(C) \geq 1/2$. If $\bar{a} \notin A_1$, then $x(\bar{a}) \leq 1/2 \leq U(C)$. If $\bar{a} \in A_1$, we distinguish two cases. 
 \begin{itemize}
 \item[(i)] If $x(a') = 0$ then $U(C) \geq 1$ and the one unit of flow on $\bar{a}$ can be rerouted.
 \item[(ii)] If $x(a') = 1/2$, then $a' \notin A_0$. Note that $C$ is a $\tail{a'}$-$t$-cut in $D$ and thus there is $a'' \in C \setminus A_1 \cup \{a'\}$ by termination condition of the while-loop. Note that, because $a'' \notin A_1$, we have $a'' \neq a$ and $x(a'') \leq 1/2$. Thus $U(C) \geq 1$ also in this last case.
 \end{itemize}
 We conclude that $x$ is indeed strictly reroutable.
\end{proof}

\begin{remark}
Note that our proof of \cref{thm:comb-alg} does not make use of \cref{thm:equivalence}. Instead, it gives a simple alternative argument for the equivalence of reroutable and strictly reroutable flows in unit capacity networks, for the special case of unit value flows.
\end{remark}

\begin{remark}\label{rem:non-half-int-example}
\cref{thm:comb-alg} implies that, for networks with $u \equiv 1$, if there exists any reroutable flow of value~$1$, then there exists a {\halfint} strictly reroutable flow of value $1$. The example given in \cref{fig:fractionality}, however, reveals that this is no longer true for flows of higher value. To see this, consider any reroutable flow of value $\Delta$ in the depicted network.
Recall that backup links can only be used for rerouting.
Thus, all nominal flow (i.e., before failure) must pass one of the three $s$-$v$-paths or the arc $(s, v)$. Let $x_1, x_2, x_3$ be the flow values on these three paths and $x^*$ be the flow value on the arc $(s, v)$. Also, all flow must pass one of the three $v$-$t$-paths. Let $x'_1, x'_2, x'_3$ be the arc flow values on these paths. We obtain $x_1 + x_2 + x_3 + x^* = x'_1 + x'_2 + x'_3 = \Delta$. We further show that $x_i + x'_j \leq 1$ for every $i, j$: To see this, consider the rerouting when the $i$th arc of the $j$th $v$-$t$-path fails. Observe that the only backup link leads to a node on the $i$th $s$-$v$-path. Hence $x_i + x'_j \leq 1$. We deduce that $3 + x^* \geq 2\Delta$. For $\Delta = 2$ this yields the unique solution $x^* = 1$, $x_1 = x_2 = x_3 = 1/3$, and $x'_1 = x'_2 = x'_3 = 2/3$, which can be verified to be a strictly reroutable flow.
Note that the example can be generalized to arbitrarily small fractional values by introducing $k$ instead of only $3$ parallel paths in each of the two segments.
\end{remark}

\section{Hardness results}\label{sec:hardness}

In this section, we give hardness results for {\maxreroute} and some variants of the problem. 

\paragraph{Paths avoiding forbidden pairs}
Our hardness results are based on reductions from
{\forbiddenpairs}, which is defined as follows:
We are given a digraph $D' = (V', A')$, two nodes $s', t' \in V'$, and a set of forbidden arc pairs $\mathcal{F} \subseteq \{\{a, \bar{a}\} \suchthat a, \bar{a} \in A\}$. The task is to find an $s'$-$t'$-path $P$ that does not contain both arcs of any pair, i.e., $|S \cap P| \leq 1$ for all $S \in \mathcal{F}$. It is not hard to see that {\forbiddenpairs} is \NP-hard~\cite{gabow1976two}.

In all reductions that follow, we will implicitly make the following three assumptions on the digraph $D' = (V', A')$ given in the {\forbiddenpairs} instance.
  \begin{enumerate}
    \item Forbidden pairs are disjoint, i.e., $S \cap S' = \emptyset$ for $S, S' \in \mathcal{F}$.
    \item If $a \in S$ for some $S \in \mathcal{F}$, then $\outarcs{\tail{a}} = \{a\}$.
    \item If $a \in S$ and $a' \in S'$ for some $S, S' \in \mathcal{F}$, then $\head{a} \neq \tail{a'}$.
  \end{enumerate}
It is easy to see that these assumptions are without loss of generality. They can be ensured by subdividing arcs, without changing the feasibility of the {\forbiddenpairs} instance.

\subsection{General capacities}\label{sec:capacity-hardness}

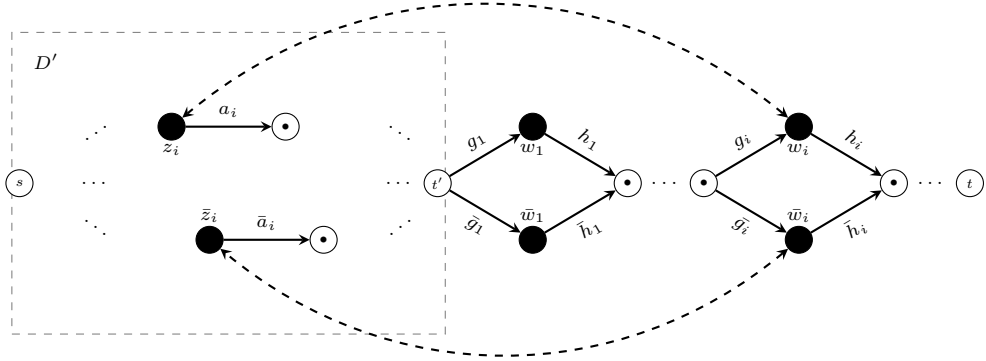
\begin{figure}[t]
  \hspace*{-0.2cm}%
  \begin{tikzpicture}[font=\scriptsize]
      \draw[black!50!white, dashed] (-0.1, 2) rectangle ++(5.7, -4);
      \node[anchor=north west, inner sep=5pt] at (0, 1.9) {$D'$};
			\path node[labeledNode] (s) {$s$}
				++(1, 0) node {\dots}
				+(0, 0.75) node {$\iddots$}
				+(0, -0.5) node {$\ddots$}
				++(1, 0.75) node[normalNode] (v1) {}
				  node[below=3pt] {$z_i$}
				++(1.5, 0) node (v2) [safeNode]
					edge[normalEdge, <-] node[above] {$a_i$} (v1)
				+(-1, -1.5) node[normalNode] (v11) {}
				  node[above=3pt] {$\bar{z}_i$}
				+(0.5, -1.5) node (v12) [safeNode]
					edge[normalEdge, <-] node[above] {$\bar{a}_i$} (v11)
				++(1.5, -0.75) node {\dots}
				+(0, 0.75) node {$\ddots$}
				+(0, -0.5) node {$\iddots$}
				++(0.5, 0) node[labeledNode] (t1) {$t'$}
				+(1.25, 0.75) node[normalNode] (w01) {}
				  edge[normalEdge, <-] node[sloped, above] {$g_1$} (t1)
				  node[below=3pt] {$w_1$}
				+(1.25, -0.75) node[normalNode] (w02) {}
				  edge[normalEdge, <-] node[sloped, below] {$\bar{g}_1$} (t1)
				  node[above=3pt] {$\bar{w}_1$}
				++(2.5, 0) node (w0t) [safeNode]
					edge[normalEdge, <-] node[sloped, above] {$h_1$} (w01)
					edge[normalEdge, <-] node[sloped, below] {$\bar{h}_1$} (w02)
				++(0.5, 0) node {$\dots$}
				++(0.5, 0) node (w) [safeNode]
				+(1.25, 0.75) node[normalNode] (w1) {}
				  edge[normalEdge, <-] node[sloped, above] {$g_i$} (w)
					edge[secondaryEdge, <->, bend right=40] (v1)
				  node[below=3pt] {$w_i$}
				+(1.25, -0.75) node[normalNode] (w2) {}
				  edge[normalEdge, <-] node[sloped, below] {$\bar{g_i}$} (w)
					edge[secondaryEdge, <->, bend left=40] (v11)
				  node[above=3pt] {$\bar{w}_i$}
				++(2.5, 0) node (wt) [safeNode]
					edge[normalEdge, <-] node[sloped, above] {$h_i$} (w1)
					edge[normalEdge, <-] node[sloped, below] {$\bar{h}_i$} (w2)
				++(0.5, 0) node {\dots}
				++(0.5, 0) node[labeledNode] (t) {$t$};
	\end{tikzpicture}

\caption{Construction for the proof of \cref{thm:general-capacities-hardness}. The dashed box contains the graph~$D'$ from the {\forbiddenpairs} instance. The arcs $a_i, \bar{a}_i$ have capacity $2$ for all $i$, all other arcs have unit capacity. In a reroutable flow of value $2$, the arcs $h_i$ and $\bar{h}_i$ must be saturated for all $i$. Any rerouting for $a_i$ has to traverse $h_i$ and any rerouting for $\bar{a}_i$ has to traverse $\bar{h}_i$.}
\label{fig:hardness-construction}
\end{figure}

\repeattheorem{thmGeneralCapacities}
\begin{proof}
We are given an instance $(D' = (V', A'), s', t', \mathcal{F})$ of {\forbiddenpairs}.
We construct an instance of {\maxreroute} as follows.
Denote the forbidden pairs in $\mathcal{F}$ by $\{a_1, \bar{a}_1\}, \dots, \{a_k, \bar{a}_k\}$ in some arbitrary order. We introduce new nodes $v_2, \dots v_{k + 1}$, $w_1, \dots, w_k$, $\bar{w}_1, \dots, \bar{w}_k$ and define $s := s'$, $v_1 := t'$, and $t := v_{k+1}$. We then introduce arcs $g_i = (v_i, w_i),\ h_i = (w_i, v_{i + 1})$ and $\bar{g}_i = (v_i, \bar{w}_i),\ \bar{h}_i = (\bar{w}_i, v_{i + 1})$ for every $i \in [k]$. We also add the arc $(s, t')$.
Furthermore, we introduce backup links from $v_{i}$ to $t$ for every $i \in [k]$ and from $v$ to $t$ for every $v \in V'$.
For every forbidden pair of arcs $\{a_i, \bar{a}_i\}$, we introduce four backup paths: from $z_i := \tail{a_i}$ to $w_i$ and vice versa and from $\bar{z}_i := \tail{\bar{a}_i}$ to $\bar{w}_i$ and vice versa. Finally, we set capacities $u(a_i) := u(\bar{a}_i) := 2$ for all $i \in [k]$ and $u(a) := 1$ for all other arcs $a$. See \cref{fig:hardness-construction} for an illustration of the construction.

We show that the {\maxreroute} instance constructed above allows for a reroutable flow of value $2$ if and only if there is an $s'$-$t'$-path avoiding the forbidden pairs in $\mathcal{F}$. In the following, we call a node $v$ \emph{safe} if there is a backup link from $v$ to $t$.

\begin{description}
\item[Sufficiency] Assume there is an $s'$-$t'$-path $P$ in $D'$ avoiding all forbidden pairs. We will construct a reroutable flow of value $2$. First, we extend $P$ to an $s$-$t$-path $Q$ as follows.
 For $i \in [k]$, attach $g_i$ and $h_i$ to $Q$ if $a_i \in P$, and attach $\bar{g}_i$ and $\bar{h}_i$ otherwise.
 As for every $i$, the path $Q$ uses either $h_i$ or $\bar{h}_i$, there is another $s$-$t$-path $\bar{Q}$ in $G$ that is arc-disjoint from $Q$ (starting with $(s, t')$ and then using the complement of $Q$ from $t'$ to $t$). We set $x(Q) := x(\bar{Q}) := 1$.
 
 We now verify that $x$ is a reroutable flow. As the nodes $w_i$, $\bar{w}_i$, $z_i$, and $\bar{z}_i$ for $i \in [k]$ are the only nodes that are not safe,  we only need to check that there is a rerouting for failure of the arcs $a_i$, $\bar{a}_i$, $h_i$ and $\bar{h}_i$.
 Because $\rcap[h_i]{x}{a_i} = 2 - x(a_i) \geq 1$, we can concatenate the backup link from $w_i$ to $z_i$ with the arc $a_i$ and the backup link from $\head{a_i}$ to $t$ to obtain a path with residual capacity $1$. This is enough to reroute the flow $x$ in case of failure of $h_i$. An identical argument applies for failure of $\bar{h}_i$.
  Further note that $a_i \in Q$ implies $h_i \in Q$ by construction of $Q$. Therefore, either $x(a_i) = 0$ or $\rcap[a_i]{x}{h_i} = 1$.
  In the former case, rerouting is trivial.
  In the latter case, we can concatenate the backup link from $z_i$ to $w_i$ with $h_i$ and the backup link from $v_{i+1}$ to $t$ to reroute the one unit of flow on $a_i$. Likewise, we observe that $\bar{a}_i \in Q$ implies $\bar{h}_i \in Q$, because $\bar{a}_i \in P$ implies $a_i \notin P$. Hence we can construct a rerouting for $\bar{a}_i$ in the same manner.

\item[Necessity] Now assume there exists a reroutable flow $x$ of value $2$. Observe that $x(h_i) = x(\bar{h}_i) = 1$ for all $i \in [k]$ (as the flow cannot use the backup links and thus has to traverse these arcs). Let $Q \in \paths$ be any path with $x(Q) > 0$ and $(s, t') \notin Q$. 
Let $P' := Q[s', t']$ be the projection of $Q$ to $D'$.
We claim that $P'$ avoids all forbidden pairs in $\mathcal{F}$. 
By contradiction assume this is not the case, i.e., $a_i, \bar{a}_i \in P' \subseteq Q$ for some $i$. 
Consider the rerouting $x_{a_i}$ for failure of $a_i$. 
Note that assumptions (i) and (ii) from the beginning of the proof imply that the only $z_i$-$t$-path in $\NetRemoveArc{a_i}$ contains $h_i$. Hence all flow in the rerouting has to traverse $h_i$, and therefore $$x(a_i) = x_{a_i}(h_i) \leq \rcap[a_i]{x}{h_i} = \elsum{P \in \pathsVia{a_i}{h_i}} x(P) \leq x(a_i).$$ 
We conclude that all the above inequalities must be fulfilled with equality, and hence $Q \in \pathsVia{a_i}{h_i}$, and in particular $h_i \in Q$. By a symmetric argument, we conclude that $\bar{a}_i \in Q$ implies $\bar{h}_i \in Q$. 
However there is no $s$-$t$-path that contains both $h_i$ and $\bar{h}_i$. This yields the desired contradiction. Hence $P'$ is a path avoiding all forbidden pairs. \qedhere
\end{description}
\end{proof}

\subsection{Computing integral flows}\label{sec:integrality}

In \cref{sec:integral-solutions} we provided an algorithm that was able to decide whether a $1/2$-integral reroutable flow of value $1$ exists in a unit capacity network. 
If however, we require the flow to be integral, this question suddenly becomes \NP-hard. Note that this problem corresponds to finding a single `reroutable' path.

\repeattheorem{thmIntegralHardness}

\begin{proof}
Let $(D' = (V', A'), s', t', \mathcal{F})$ be an instance of {\forbiddenpairs}.
We construct an instance of {\maxreroute} as follows.
We let $s := s'$ and $t := t'$.
For every forbidden pair $\{a, \bar{a}\} \in \mathcal{F}$, we add a backup path from $\tail{a}$ to $\tail{\bar{a}}$ and vice versa. For every node $v$ that is not the tail of an arc that appears in a forbidden pair, we add a backup link to $t$.
We set unit capacities to all arcs.

Let $P$ be an $s$-$t$-path path avoiding all forbidden pairs. Then the flow $x$ with $x(P) = 1$ is a reroutable flow. To see this, observe that the only nodes that are not safe are the tails of arcs participating in forbidden pairs. Consider a forbidden pair $\{a, \bar{a}\} \in \mathcal{F}$. Observe that $x(a) = 0$ or $x(\bar{a}) = 0$ and hence a rerouting exists in case of failure of either of the two arcs.

Now assume there is an integral reroutable flow of value $1$. This means that $x(P) = 1$ for an $s$-$t$-path $P$ that does not use any backup link. Assume that $a, \bar{a} \in P$ for some $\{a, \bar{a}\} \in \mathcal{F}$. W.l.o.g. assume $a <_P \bar{a}$. Note that $\{a, \bar{a}\}$ is a $\tail{\bar{a}}$-$t$-cut in $D$. Hence, when $\bar{a}$ fails, the rerouting must use $a$, i. e., $x_{\bar{a}}(a) = 1 > 0 = \rcap[\bar{a}]{x}{a}$, a contradiction.

Because of \cref{thm:equivalence} and Remark~\ref{rem:integral-equivalence} the constructed instance allows for an integral strictly reroutable flow of value $1$ if and only if it allows for a reroutable flow of value $1$. Hence the hardness result carries over to {\maxstrictreroute}.
\end{proof}

\subsection{Multiple arc failures}\label{sec:multiple-failures}

A natural generalization of {\maxreroute} and {\maxstrictreroute} allows multiple simultaneous arc failures.
When a set of arcs $S$ fails, flow is interrupted where it first encounters an arc from $S$ and has to be rerouted from that point to the sink.

To formalize this, we introduce the following notation.
For $S \subseteq A$ and $a \in A$ we define $\paths[S] := \bigcup_{\bar{a} \in S} \paths[\bar{a}]$ and $\pathsVia{S}{a} := \bigcup_{\bar{a} \in S} \pathsVia{\bar{a}}{a}$.
Let $x$ be an $s$-$t$-flow.
We extend the notion of capacity available after failure of $S$ by defining
$$\rcap[S]{x}{a} := u(a) - \elsum{\quad\ P \in \paths[a] \setminus \pathsVia{S}{a}} x(P).$$
A rerouting of $x$ for failure of $S$ consists of a collection of flows $(x_{S,\bar{a}})_{\bar{a} \in S}$ in $\NetRemoveSet{S}$, such that $x_{S,\bar{a}}$ is a $\tail{\bar{a}}$-$t$-flow for each $\bar{a} \in S$, with 
$$\val{x_{S,\bar{a}}} = \elsum{\quad\ P \in \paths[\bar{a}] \setminus \pathsVia{S}{\bar{a}}} x(P) \quad \text{ and }\quad \sum_{\bar{a} \in S} x_{S,\bar{a}}(a) \leq \rcap[S]{x}{a} \text{ for all $a \in A \setminus S$.}$$
A rerouting is strict if $\sum_{\bar{a} \in S} x_{S,\bar{a}}(a) \leq \rcap{x}{a}$ for all $a \in A$.
A flow is (strictly) $k$-reroutable if it has a (strict) rerouting for failure of any set $S \subseteq A$ with $|S| \leq k$.
We denote the corresponding problem of finding a (strictly) $k$-reroutable flow of maximum value by \textsc{Max (Strictly) $k$-Reroutable Flow}.
It turns out that dealing even with only $2$ arc failures in unit capacity networks is \NP-hard in both cases. 

\begin{theorem}\label{thm:MultiFailure}
\textsc{Max (Strictly) $k$-Reroutable Flow} is \NP-hard, even when restricted to instances with $k=2$ and $u \equiv 1$.
\end{theorem}

\begin{proof}
Let $(D' = (V', A'), s', t', \mathcal{F})$ be an instance of {\forbiddenpairs}.
We construct an instance $(D = (V, A), s, t, u)$ of \textsc{Max $2$-Reroutable Flow} as follows. 
Let $\{a_1, \bar{a}_1\}, \dots, \{a_{\ell}, \bar{a}_{\ell}\}$ be the forbidden pairs in $\mathcal{F}$ in some arbitrary order.
For $i \in [\ell]$ we add vertices $v_i, w_i, v'_i, w'_i$ and $\bar{w}_i, \bar{v}'_i, \bar{w}'_i$, and define $v_{\ell + 1} := s'$.
We add arcs $g_i = (v_i, w_i), h_i = (w_i, v'_i), g'_i = (v'_i, w'_i), h'_i = (w'_i, v_{i + 1})$ as well as $\bar{g}_i = (v_i, \bar{w}_i), \bar{h}_i = (\bar{w}_i, \bar{v}'_i), \bar{g}'_i = (\bar{v}'_i, \bar{w}'_i), \bar{h}'_i = (\bar{w}'_i, v_{i + 1})$.
We generalize the concept of backup links to $2$-backup links, which are simply two parallel internally node-disjoint backup links.
We add $2$-backup links from each of $v_i$, $v'_i$, and $\bar{v}'_i$ to $t'$, thus making them safe nodes.
We further add $2$-backup links $w_i$ to $w'_i$, from $w'_i$ to $z_i$, and from $z_i$ to $w_i$, as well as from $\bar{w}_i$ to $\bar{w}'_i$, from $\bar{w}'_i$ to $\bar{z}_i$, and from $\bar{z}_i$ to $\bar{w}_i$.
For every node $v \in V'$ that is not the tail of an arc that participates in a forbidden pair, we add a $2$-backup link from $v$ to $t$.
We also add an additional $s'$-$t'$-path $\bar{P}$ of length $6$ consisting of the arcs $e_1, \dots, e_6$. We add $2$-backup links from $\tail{e_2}$ to $\tail{e_4}$, from $\tail{e_4}$ to $\tail{e_6}$, and from $\tail{e_6}$ to $\tail{e_2}$. We also add $2$-backup links from each of $\tail{e_1}$, $\tail{e_3}$ and $\tail{e_5}$ to $t'$.
Finally, we define $s := v_1$, $t := t'$, and $u \equiv 1$.

We show two implications: (i) If there is an $s'$-$t'$-path in $D'$ avoiding all pairs, then there is a strictly reroutable flow of value $1$ in the constructed network $D$. (ii) If there is a reroutable flow of value $1$ in the constructed network $D$, then there is an $s'$-$t'$-path in $D'$ avoiding all pairs. Note that (i) and (ii) together also imply that a strictly reroutable flow of value $1$ exists in $D$ if and only if there is a reroutable flow of value $1$. Therefore, both \textsc{Maximum $k$-Reroutable Flow} and \textsc{Maximum Strictly $k$-Reroutable Flow} can be used to solve {\forbiddenpairs}.

\begin{itemize}
\item[(i)]
Let $P'$ be an $s'$-$t'$ path in $D'$ avoiding all forbidden pairs. Let $O$ be the $s$-$s'$-path defined by the following rule: For $i \in [\ell]$, if $a_i \in P'$, then use the segment $g_i, h_i, g'_i, h'_i$, otherwise use the segment $\bar{g}_i, \bar{h}_i, \bar{g}'_i, \bar{h}'_i$. Let $Q$ be the concatenation of $O$ and $P'$ to an $s$-$t$-path. Further, let $\bar{O}$ be the $s$-$s'$-path that uses the segments not used by $O$ and let $\bar{Q}$ be the concatenation of $\bar{O}$ and $\bar{Q}$.
We define the flow $x(Q) := 1/2$ and $x(\bar{Q}) := 1/2$.

Consider the failure of any set $S \subseteq A$ with $|S| = 2$. Note that $\rcap{x}{a} \geq 1/2$ for all $a \in A$. It is also easy to see that for any $\bar{a} \in A$ there are three disjoint $\tail{\bar{a}}$-$t$-paths in $D$ and hence at least one such path in $\NetRemoveSet{S}$. If $Q \cap S = \emptyset$ or $\bar{Q} \cap S = \emptyset$, we have to reroute at most $1/2$ unit of flow, which is possible by the above observations. Hence, consider the case that $S = \{q_1, q_2\}$ with $q_1 \in Q$ and $q_2 \in \bar{Q}$. If there is no rerouting for $S$, then there must be an arc $q_3$ such that $\{q_1,q_2,q_3\}$ separate $\tail{q_1}$ and $\tail{q_2}$ from $t$. It is easy to see that this is only possible if $\{q_1, q_2, q_3\} = \{a_i, h_i, h'_i\}$ or $\{q_1, q_2, q_3\} = \{\bar{a}_i, \bar{h}_i, \bar{h}'_i\}$ for some $i \in [\ell]$---w.l.o.g., we assume the former.
Note that $a_i \notin \bar{Q}$ and that $h_i, h'_i$ are both on the same path. But $h_i, h'_i \in \bar{Q}$ implies $a_i \notin Q$ by construction of $Q$ and the fact that $P'$ avoids forbidden pairs. Hence there must be a rerouting for $x$.

\item[(ii)] Let $x$ be a reroutable flow of value $1$. We first observe that for every $i \in [\ell]$, $x(h_i) = x(h'_i) = x(\bar{h}_i) = x(\bar{h}'_i) = 1/2$. To see this, note that the nominal flow cannot use backup links and hence $x(h_i) + x(\bar{h}'_i) = 1$. Now assume that $x(h_i) > 1/2$. Then consider the failure of $S_i := \{h'_i, a_i\}$. Note that all flow on $h'_i$ and $a_i$ must be rerouted via $h_i$ in this case, i.e.,
\begin{align}
\elsum{P \in \paths[S_i]} x(P) \ \leq \ \rcapEq[S_i]{x}{h_i} \ = \ 1 - x(h_i).\label{eq:failureS}
\end{align}

But $x(h'_i) = x(h_i) > 1/2 > \rcap[\{h'_i, a_i\}]{x}{a}$, a contradiction. Analogously, we can see that $x(\bar{h}_i) > 1/2$ is not possible. By similar arguments, namely the failure of $\{e_4, e_6\}$ we can observe that $x(e_1) = \dots = x(e_6) \leq 1/2$. Hence there must be an $s$-$t$-path $Q$ in $D$ with $x(Q) > 0$ and $e_1, \dots, e_6 \notin Q$. Consider the suffix $P' := Q[s', t']$. By contradiction assume $P'$ contains a forbidden pair of arcs $a_i, \bar{a}_i$ for some $i \in [\ell]$. Again consider the failure of $\{h'_i, a_i\}$. By \eqref{eq:failureS}, we obtain $\sum_{P \in \paths[S_i]} x(P) \leq 1/2$. Since $x(h'_i) = 1/2$, this implies that $a_i \in P$ implies $h'_i \in P$ for all $P$ with $x(P) > 0$. In particular, $h'_i \in Q$. By a symmetric argument we obtain $\bar{h}'_i \in Q$. But $h'_i, \bar{h}'_i \in Q$ is not possible by construction of the digraph $D$. Thus we arrived at the desired contradiction. \qedhere
\end{itemize}
\end{proof}

\paragraph{Acknowledgments} We thank David Adjiashvili and Marco Senatore for helpful discussions. This work was supported by the Alexander von Humboldt Foundation with funds of the German Federal Ministry of Education and Research (BMBF) and by an NSERC Discovery Grant.

\bibliographystyle{plain}
\bibliography{references}

\end{document}